\documentclass[journal, twocolumn,10pt]{IEEEtran}
\IEEEoverridecommandlockouts
\usepackage{amsmath}
\usepackage{amsfonts}
\usepackage{amssymb}
\usepackage{amsthm}
\usepackage{graphicx}
\usepackage{psfrag}
\usepackage{cite}
\usepackage{algorithm}
\usepackage{algorithmic}

\usepackage{color}
\usepackage{float}%
\usepackage{epsfig,array,multicol,verbatim}
\usepackage{subfigure}
\usepackage{cases}
\usepackage{setspace}
\usepackage{epstopdf}
\usepackage{nomencl} 

\theoremstyle{plain} \newtheorem{theorem}{Theorem}
\theoremstyle{plain} \newtheorem{proposition}{Proposition}
\theoremstyle{plain} 
\theoremstyle{remark} \newtheorem{remark}{Remark}
\theoremstyle{remark} \newtheorem{lemma}{Lemma}
\theoremstyle{plain} 
\theoremstyle{plain} 

\makenomenclature   

\begin{document}
\title{Profit-Maximizing Planning and Control of Battery Energy Storage Systems for Primary Frequency Control}
\author{Ying~Jun~(Angela)~Zhang,~\IEEEmembership{Senior Member,~IEEE}, Changhong~Zhao,~\IEEEmembership{Member,~IEEE}, Wanrong~Tang,~\IEEEmembership{Student Member,~IEEE}, Steven~H.~Low,~\IEEEmembership{Fellow,~IEEE},
\thanks{This work was supported in part by the National Basic Research Program (973 program Program number 2013CB336701), and three grants from the Research Grants Council of Hong Kong under General Research Funding (Project number 2150828 and 2150876) and Theme-Based Research Scheme (Project number T23-407/13-N).}
\thanks{Y. J. Zhang and W. Tang are with the Department of Information Engineering, The Chinese University of Hong Kong, Hong Kong. They are also with Shenzhen Research Institute, The Chinese University of Hong Kong}
\thanks{C. Zhao and S. Low are with the Engineering and Applied Science Division, California Institute of Technology, Pasadena, CA, 91125, USA.}}

\maketitle

\begin{abstract}
We consider a two-level profit-maximizing strategy, including planning and control, for battery energy storage system (BESS) owners that participate in the primary frequency control (PFC) market. Specifically, the optimal BESS control minimizes the operating cost by keeping the state of charge (SoC) in an optimal range. Through rigorous analysis, we prove that the optimal BESS control is a ``state-invariant" strategy in the sense that the optimal SoC range does not vary with the state of the system. As such, the optimal control strategy can be computed offline once and for all with very low complexity. Regarding the BESS planning, we prove that the the minimum operating cost is a decreasing convex function of the BESS energy capacity. This leads to the optimal BESS sizing that strikes a balance between the capital investment and operating cost. Our work here provides a useful theoretical framework for understanding the planning and control strategies that maximize the economic benefits of BESSs in ancillary service markets.
\end{abstract}

\nomenclature{$ c_e $}{ electricity purchasing and selling price }
\nomenclature{$ c_p $}{ penalty rate of PFC energy shortage }
\nomenclature{$ n $}{ interval index }
\nomenclature{$ I_n $}{ $n^{th}$ $I$ interval or length of the $n^{th}$ $I$ interval }
\nomenclature{$ J_n $}{ $n^{th}$ $J$ interval or length of the $n^{th}$ $J$ interval }
\nomenclature{$ t_n^s $}{ start time of the $n^{th}$ $I$ interval }
\nomenclature{$ t_n^e $}{ end time of the $n^{th}$ $I$ interval }
\nomenclature{$ q_n $}{ indicator variables of $n^{th}$ frequency excursion event }
\nomenclature{$ {P_{PFC}}_n $}{ regulation power for $n^{th}$ $J$ interval }
\nomenclature{$ p(t) $}{ battery charging and discharging power at time $t$}
\nomenclature{$ p_{ac} (t) $}{ power exchanged with the AC bus at time $t$}
\nomenclature{$ \eta $}{ charging and discharging efficiency }
\nomenclature{$ E_{max} $}{ battery capacity  }
\nomenclature{$ P_{max} $}{ maximum charging power of battery  }
\nomenclature{$ s_n $}{ SoC of the BESS at the beginning of the $n^{th}$ $I$ interval}
\nomenclature{$ s_n^e $}{ SoC of the BESS at the end of the $n^{th}$ $I$ interval}
\nomenclature{$ cost_{e,n} $}{ battery charging cost incurred in the $n^{th}$ $I$ interval  }
\nomenclature{$ cost_{p,n}(s_n^e) $}{ penalty assessed in the $n^{th}$ $J$ interval }
\nomenclature{$ H^*_n(s_n) $}{ optimal value at the $n^{th}$ stage of the multi-stage problem }
\nomenclature{$ \alpha $}{ discounting factor  }
\printnomenclature 

\section*{Nomenclature}

\begin{tabular}{l l }
  $ c_e $ & electricity purchasing and selling price \\
  $ c_p $ & penalty rate for PFC regulation failure \\
  $ n $ & interval index \\
  $ I_n $ & length of the $n^{th}$ $I$ interval\\
  $ J_n $ & length of the $n^{th}$ $J$ interval\\
  $ t_n^s $ & start time of the $n^{th}$ $I$ interval \\
  $ t_n^e $ & end time of the $n^{th}$ $I$ interval \\
  $ q_n $ & indicator variable of $n^{th}$ excursion event \\
  $ {P_{PFC}}_n $ & PFC power requested in the $n^{th}$ $J$ interval \\
  $ p(t) $ & battery charging/discharging power at time $t$ \\
  $ p_{ac} (t) $ & power exchanged with the AC bus at time $t$ \\
  $ \eta $ & battery charging and discharging efficiency \\
  $ E_{max} $ & battery capacity \\
  $ P_{max} $ & maximum charging power of battery \\
  $ s_n $ & SoC at the beginning of the $n^{th}$ $I$ interval \\
  $ s_n^e $ & SoC at the end of the $n^{th}$ $I$ interval \\
  $ cost_{e,n} $ & charging cost incurred in the $n^{th}$ $I$ interval \\
  $ cost_{p,n} $ & penalty assessed in the $n^{th}$ $J$ interval \\
\end{tabular}

\section{Introduction}


The instantaneous supply of electricity in a power system must match the time-varying demand as closely as possible. Or else, the system frequency would rise or decline, compromising the power quality and security. To ensure a
stable frequency at its nominal value, the Transmission System Operator (TSO) must keep control reserves compensate for unforeseen mismatches between generation and load. Frequency control is performed in three levels, namely primary, secondary, and tertiary controls \cite{ kundur1994power}. The first level, primary frequency control (PFC), reacts within the first few seconds when system frequency falls outside a dead band, and restores quickly the balance between the active power generation and consumption.
Due to its stringent requirement on the response time, PFC is the most expensive control reserve. This is because PFC is traditionally performed by thermal generators, which are designed to deliver bulk energy, but not for the provision of fast-acting reserves. To complement the generation-side PFC, load-side PFC has been considered as a fast-responding and cost-effective alternative \cite{schweppe1980homeostatic, zhao2014design,zhao2013optimal,short2007stabilization, molina2011decentralized}. Nonetheless, the provision of load-side PFC is constrained by end-use disutility caused by load curtailment.

Battery energy storage systems (BESSs) have recently been advocated as excellent candidates for PFC due to their extremely fast ramp rate\cite{leecoordinated,xu2014bess}. Indeed, the supply of PFC reserve has been identified as the highest-value application of BESSs \cite{oudalov2006value}. According to a 2010 NREL report \cite{denholm2010role}, the annual profit of energy storage devices that provide PFC reserve is as high as US\$236-US\$439 per KW in the U.S. electricity market.
The use of BESS as a frequency control reserve in island power systems dates back to about 20 years ago \cite{kottick1993battery}. Due to the fast penetration of renewable energy sources, the topic recently regained research interests in both interconnected power systems \cite{ borsche2013power, xu2014bess} and microgrids \cite{mercier2009optimizing, aghamohammadi2014new}. 

In view of the emerging load-side PFC markets instituted worldwide \cite{spees2007demand,FERC}, we are interested in deriving profit-maximizing planning and control strategies for BESSs that participate in the PFC market. In particular, the optimal BESS control aims to minimize the operating cost by scheduling the charging and discharging of the BESS to keep its state of charge (SoC) in a proper range. Here, the operating cost includes both the battery charging/discharging cost and the penalty cost when the BESS fails to provide the PFC service according to the contract with the TSO. We also determine the optimal BESS energy capacity that balances the capital cost and the operating cost.
Previously, \cite{ mercier2009optimizing, xu2014bess} investigated the problem of BESS dimensioning and control, with the aim of maximizing the profit of BESS owners. There, the BESS is charged or discharged even when system frequency is within the dead band to adjust the state of charge (SoC). This is to make sure that the BESS has enough capacity to absorb or supply power when the system frequency falls outside the dead band. A different approach to correct the SoC was proposed in \cite{borsche2013power}, where the set point is adjusted to force the frequency control signal to be zero-mean.

To complement most of the previous work based on simulations or experiments, we develop a theoretical framework for analyzing the optimal BESS planning and control strategy in PFC markets. In particular, the optimal BESS control problem is formulated as a stochastic dynamic program with continuous state space and action space. Moreover, the optimal BESS planning problem is derived by analyzing the optimal value of the dynamic programming, which is a function of the BESS energy capacity. A key challenge here is that the complexity of solving a dynamic programming problem with continuous state and action spaces is generally very high. Moreover, standard numerical methods to solve the problem do not reveal the underlying relationship between the operating cost and the energy capacity of
the BESS. Our main contributions in addressing this challenge are summarized as follows.

\begin{itemize}
\item We prove that with slow-varying electricity price, the optimal BESS control problem reduces to finding an optimal target SoC every time the system frequency falls inside the dead band. In other words, the optimal decision can be described by a scalar, and hence the dimension of the action space is greatly reduced.
\item We show that the optimal target SoC is a range that is \emph{invariant} with respect to the system state at each stage of the dynamic programming. Moreover, the range reduces to a fixed point either when the battery charging/discharging efficiency approaches 1 or when the electricity price is much lower than the penalty rate for regulation failure. This result is extremely appealing, for the optimal target SoC can be calculated offline once and for all with very low complexity.
\item We prove that the minimum operating cost is a decreasing convex function of the BESS energy capacity. Based on the result, we discuss the optimal BESS planning strategy that strikes a balance between the capital cost and the operating cost.
\end{itemize}

The rest of the paper is organized as follows. In Section \ref{sec:system}, we describe the system model. The BESS operation problem is formulated as a stochastic dynamic programming problem in Section \ref{sec:formulation}. In Section \ref{sec:optimalSoC}, we derive the optimal BESS operation strategy, which is a range of target SoC \emph{independent} of the system state. The optimal BESS planning is discussed in Section \ref{sec:optimalPlan}. Numerical results are presented in Section \ref{sec:numerical}. Finally, the paper is concluded in Section \ref{sec:conclusions}.

\section{System Model}\label{sec:system}

We consider a profit-seeking BESS selling PFC service in the ancillary service market. The BESS receives remuneration from the TSO for providing PFC regulation, and is liable to a penalty whenever the BESS fails to deliver the service as specified in the contract with the TSO. We endeavour to find the optimal planning and control of the BESS to maximize its profit in the PFC market.

\subsection{System Timeline}
Most of the time, the system frequency stays inside a dead band (typically 0.04\%) centred around the nominal frequency. Once the system frequency falls outside the dead band, the TSO sends regulation signals to regulating units, including the BESS. The BESS needs to supply power (i.e., be discharged) in a frequency under-excursion event and absorb power (i.e., be charged) in a frequency over-excursion event.

The system time can be divided into two types of intervals as illustrated in Fig. \ref{fig:timeline}. The $I$ intervals are the ones during which PFC is not needed, i.e., when the system frequency stays inside the dead band or when the frequency is regulated by secondary or tertiary reserves. An $I$ interval ends and a $J$ interval starts, when a frequency excursion event occurs. The lengths of the $J$ intervals are the PFC deployment times requested by the TSO.

The lengths of the $n^{th}$ $I$ and $J$ intervals are denoted as $I_n$ and $J_n$, respectively. Suppose that $I_n$'s  are independently and identically distributed (i.i.d.) with probability density function (PDF) $f_I(x)$ and complimentary cumulative distribution function (CCDF) $\tilde{F}_I(x)$. Likewise, $J_n$'s are i.i.d. with PDF $f_J(x)$ and CCDF $\tilde{F}_J(x)$. Note that $f_I(x)=-\frac{d\tilde{F}_I(x)}{dx}$ and $f_J(x)=-\frac{d\tilde{F}_J(x)}{dx}$. Moreover, define indicator variables $q_n$ such that $q_n=1$ and $-1$ when the $n^{th}$ frequency excursion event is an over-excursion event and under-excursion event, respectively. Let $p_1=\Pr\{ q_n=1\}$ and  $p_{-1}=1-p_1=\Pr\{ q_n=-1\}$.

 \begin{figure}
\centering
\includegraphics[width=0.45\textwidth]{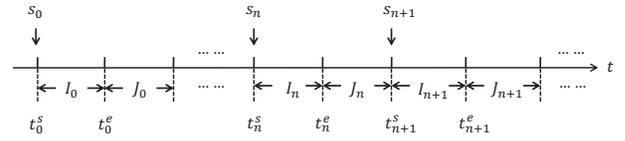}
\caption{System time line.}\label{fig:timeline}
\end{figure}

\subsection{BESS Operation}

Suppose that the  BESS has an energy capacity $E_{max}$ (kWh) and maximum charging and discharging power limits $P_{max}$ (kW). The charging and discharging efficiency is $0<\eta\leq 1$. Moreover, let $e(t)$ denote the amount of energy stored in the battery at time $t$, and $p(t)$ denote the battery charging ($p(t)>0$) or discharging ($p(t)<0$) power at time $t$. Due to the charging and discharging efficiency $\eta$, the power exchanged with the AC bus, denoted by $p_{ac}(t)$, is
\begin{equation}\label{eqn:p_ac}
p_{ac}(t)=
\begin{cases}
p(t)/\eta & \text{if}~p(t)>0 \\
p(t)\eta & \text{if}~p(t)<0
\end{cases}.
\end{equation}

In the $n^{th}$ frequency excursion event, the BESS is obliged to supply or absorb $P_{PFC,n}$ kW regulation power for the entire period of $J_n$. Here, $P_{PFC,n}$'s are i.i.d. random variables with pdf $f_{P_{PFC}}(x)$ and CCDF $\tilde{F}_{P_{PFC}}(x)$. Typically, $P_{PFC,n}$ takes value in $[0, R]$, where $R$ is the standby reserve capacity specified in the contract with the TSO. In return, the BESS is paid for the availability of the standby reserve. That is, the remuneration is proportional to $R$ and the tendering period, but independent of the actual amount of PFC energy supplied or consumed.

Let $s_n$ and $s_n^e$ denote the SoC (normalized the energy capacity $E_{max}$)  \footnote{SoC at time $t$ is defined as $s(t)=\frac{e(t)}{E_{max}}$. Obviously, $s(t)\in[0, 1]$.} of the BESS at the beginning and end of $I_n$, respectively. Obviously, when $s_n^e$ is too low or too high, the BESS may fail to supply or absorb the amount PFC energy requested by the TSO in the subsequent $J_n$ interval, resulting in a regulation failure. In this case, the BESS is assessed a penalty that is proportional to the shortage of PFC energy. Let $c_p$ be the penalty rate per kWh PFC energy shortage. Then, the penalty assessed in the $n^{th}$ frequency excursion event is
\begin{equation}\label{eqn:2}
\resizebox{\hsize}{!}{$cost_{p,n}(s_n^e)=
\begin{cases}
c_p \left(E_{PFC,n} - \frac{E_{max}(1-s_n^e)}{\eta} \right)^+ &  \text{if}~q_n=1 \\
c_p \left(E_{PFC,n}- \eta E_{max}s_n^e \right)^+ &  \text{if}~ q_n= -1
\end{cases},$}
\end{equation}
where $(x)^+=\max(x, 0)$ and $E_{PFC,n}=P_{PFC,n}J_n$ is an auxiliary variable indicating the PFC energy supplied or absorbed during $J_n$. Since $P_{PFC,n}$'s and $J_n$'s are i.i.d., respectively, $E_{PFC,n}$ are also i.i.d. variables with PDF $f_{E_{PFC}}(x)$ and CCDF  $\tilde{F}_{E_{PFC}}(x)$. Due to the battery charging/discharging efficiency, $\eta E_{PFC,n}$ and $\frac{E_{PFC,n}}{\eta}$ are the energy charged to or discharged from the BESS during the PFC deployment time.

To avoid penalty, the BESS must be charged or discharged during $I$ intervals to maintain a proper level of SoC. Suppose that the electricity purchasing and selling price, denoted by $c_e$, varies at a much slower time scale (i.e., hours) than that at which the PFC operates (i.e., seconds to minutes), and thus can be regarded as a constant during the period of interest. Then, the battery charging cost incurred in $I_n$ is calculated as
\begin{equation}\label{eqn:charging_cost}
cost_{e,n}= c_e \int_{t_n^s}^{t_n^s+I_n} p_{ac}(t) dt,
\end{equation}
where $p_{ac}$ is given in \eqref{eqn:p_ac}. $cost_{e,n}>0$ corresponds to a cost due to power purchasing, and $cost_{e,n}<0$ corresponds to a revenue due to power selling. Notice that the BESS SoC is bounded between 0 and 1. Thus, $p(t)$ is subject to the following constraint
\begin{equation}\label{eqn:BESScons}
 0\leq s_n E_{max}+\int_{t^s_n}^\tau p(t) dt \leq E_{max} ~\forall \tau\in[t^s_n, t^e_n],
\end{equation}
where $t^s_n$ and $t^e_n$ are the starting and end times of $I_n$, respectively. As a result, $s_n$ and $s_n^e$ are related as
\begin{equation}\label{eqn:6}
s_n^e=\frac{s_nE_{max}+\int_{t_n^s}^{t_n^s+I_n} p(t) dt}{E_{max}},
\end{equation}
subject to the constraint in \eqref{eqn:BESScons}. Likewise, the SoC at the beginning of the next $I$ interval, $s_{n+1}$, is related to $s_n^e$ as
\begin{equation}\label{eqn:7}
s_{n+1}=\left[ \frac{s_n^e E_{max}+ (\mathbf{1}_{q=1}\eta - \mathbf{1}_{q=-1}\frac{1}{\eta} ) E_{PFC}}{E_{max}} \right]_0^1,
\end{equation}
where $[x]_0^1=\min(1, \max(0, x))$ and $\mathbf{1}_A$ is an indicator function that equals 1 when $A$ is true and 0 otherwise.

\section{Problem Formulation}\label{sec:formulation}

As mentioned in the previous section, the remuneration the BESS receives from the TSO is proportional to the standby reserve capacity $R$ and the tendering period, but independent of the actual amount of PFC energy supplied or absorbed. With fixed remuneration, the problem of profit maximization is equivalent to the one that minimizes the capital and operating costs. In this section, we formulate the optimal BESS control problem that minimizes the operating cost $\sum_n (cost_{e,n}+cost_{p,n})$ for a given BESS capacity $E_{max}$. The optimal BESS planning problem that finds the optimal $E_{max}$ will be discussed later in Section \ref{sec:optimalPlan}.

At the beginning of each interval $I_n$, the optimal $p(t)$ during this $I$ interval is determined based on the observation of $s_n$.  When making the decision, the BESS has no prior knowledge of the realizations of $I_k$, $E_{PFC,k}$, and $q_k$ for $k=n, n+1, \cdots$.  As such, the problem is formulated as the following stochastic dynamic programming, where $s_n$ is regarded as the system state at the $n^{th}$ stage, and the state transition from $s_n$ to $s_{n+1}$ is determined by the decision $p(t)$ as well the exogenous variables $I_n$, $E_{PFC, n}$, and $q_n$.

At stage $n$, solve

\begin{equation}\label{eqn:multistage}
\resizebox{\hsize}{!}{$
\begin{aligned}
 &H^*_n(s_n)=\min_{p(t), t\in[t_n^s, t_n^e]} \mathrm{E}_{I_n, E_{PFC,n}, q_n} \left[cost_{e,n} + cost_{p,n}\left(s_n^e\right) \right] \\
&+  \alpha \mathrm{E}_{I_n, E_{PFC,n}, q_n}\left[H^*_{n+1}(g(s_{n}, p(t), I_n, E_{PFC, n}, q_n))\right] \\
s.t.
&~ \eqref{eqn:p_ac}, \eqref{eqn:BESScons}, \text{and}  \\
& -P_{max} \leq p(t)\leq P_{max} ~\forall t\in[t^s_n, t^e_n],
\end{aligned}$}
\end{equation}
where $cost_{p,n}$, $cost_{e,n}$ and $s_n^e$ are defined in \eqref{eqn:2}, \eqref{eqn:charging_cost}, and \eqref{eqn:6}, respectively. $H^*_n(s_n)$ is the optimal value at the $n^{th}$ stage of the multi-stage problem, $\alpha \in (0, 1)$ is a discounting factor, and
$g(s_{n}, p(t), I_n, E_{PFC,n}, q_n):=s_{n+1}$  describes the state transition given by \eqref{eqn:6} and \eqref{eqn:7}.

In practice, the tendering period of the service contract signed with the TSO (in the order of months) is much longer than the duration of one stage in the above formulation (in the order of seconds or minutes).  Moreover, the distributions of $I_n$, $E_{PFC, n}$, and $q_n$ are i.i.d. Thus, Problem \eqref{eqn:multistage} can be regarded as an infinite-horizon dynamic programming problem with stationary policy. In other words, the subscripts $n$ and $n+1$ in \eqref{eqn:multistage} can be removed.

Problem \eqref{eqn:multistage} requires the optimization of a continuous time function $p(t)$. When the electricity price $c_e$ remains constant within an $I$ period,  there always exists an optimal solution where battery is always charged or discharged at the full rate $P_{max}$ until a prescribed SoC target has been reached or the $I$ interval has ended. Then, finding the optimal charging/discharging policy is equivalent to finding an optimal target SoC $\pi\in[0, 1]$. This is because charging/discharging cost during an $I$ period is only related to the total energy charged or discharged, regardless of when and how fast the charging or discharging is.

 Under the full-rate policy, the battery charges/discharges at a rate $P_{max}$ until the target SoC has been reached or the $I$ interval has ended. Thus, the charging cost  \eqref{eqn:charging_cost} during $I_n$ is equal to the following, where $\pi$ is the target SoC.
 \begin{eqnarray}
&& cost_{e}(s_n, \pi)  =\\
&&\begin{cases}
\frac{c_e}{\eta}\min(P_{max} I_n,(\pi-s_n)E_{max}) & \text{if}~s_n< \pi \\ \nonumber
-c_e{\eta}\min(P_{max} I_n, (s_n-\pi)E_{max}) & \text{if}~ s_n>\pi \\
0 & \text{if}~ s_n=\pi
\end{cases}.
 \end{eqnarray}
Likewise, \eqref{eqn:6} can be written as a function of $s_n$ and $\pi$:
\begin{eqnarray}
s_n^e(s_n, \pi) = s_n+\mathrm{sgn}(\pi-s_n)\min\left( \frac{P_{max} I_n}{E_{max}}, |\pi-s_n| \right),
\end{eqnarray}
 where $\mathrm{sgn}(\cdot)$ is the sign function.

We are now ready to rewrite Problem \eqref{eqn:multistage} into the following Bellman's equation, where subscript $n$ is omitted because the problem is an infinite-horizon problem with stationary policy.
\begin{equation}\label{eqn:bellman}
\resizebox{\hsize}{!}{$H^*(s)=\min_{\pi\in[0,1]} h(s, \pi) + \alpha \mathrm{E}_{I, E_{PFC}, q}\left[ H^*(g(s, \pi, I, E_{PFC}, q))\right],$}
\end{equation}
where
\begin{equation}\label{eqn:1stage}
h(s, \pi)= \mathrm{E}_I\left[cost_e(s, \pi)\right] +  \mathrm{E}_{I, E_{PFC}, q} \left[ cost_p\left( s^e(s, \pi)\right)\right]
\end{equation}
 is the expected one stage cost. With a slight abuse of notation, define
 \begin{eqnarray}
&&g(s, \pi, I, E_{PFC}, q)  = \nonumber \\
&& \resizebox{0.9\hsize}{!}{$\left [\frac{1}{E_{max}}\left(\begin{array}{ll}
 s E_{max}+\mathrm{sgn}(\pi-s) \min\left(P_{max} I, |\pi-s|E_{max}\right)\\
  + (\mathbf{1}_{q=1}\eta - \mathbf{1}_{q=-1}\frac{1}{\eta} ) E_{PFC}
  \end{array}\right)\right]_0^1$} \nonumber
 \end{eqnarray}
as the state transition. More specifically, in \eqref{eqn:1stage}
\begin{eqnarray}\label{eqn:12}
&& \mathrm{E}_I\left[cost_e(s, \pi)\right] \\ \nonumber
& =& \left( \mathbf{1}_{\pi> s}\frac{1}{\eta} - \mathbf{1}_{\pi<s}\eta \right)c_e \times \\
&&\left( \int_0^{Q_1} P_{max} x f_I(x) dx + |\pi-s|E_{max} \tilde{F}_I\left(Q_1 \right) \right), \nonumber
\end{eqnarray}
where
\begin{equation}\label{eqn:q1}
Q_1=\frac{|\pi-s|E_{max}}{P_{max}}
\end{equation}
is the minimum time to charge or discharge the battery from $s$, the initial SOC at this stage, to $\pi$, the target SOC.
Likewise,
\begin{eqnarray}\label{eqn:14}
&&\mathrm{E}_{I, E_{PFC}, q} \left[ cost_p\left( s^e\right)\right] = \mathrm{E}_{I}\left[\overline{cost_p}(s^e) \right] \\
&=&\resizebox{0.85\hsize}{!}{$ \begin{cases}
 \int_0^{Q_1}\overline{cost_p}\left(s+\frac{P_{max} x}{E_{max}}\right)f_I(x)dx + \overline{cost_p}(\pi) \tilde{F}_I\left(Q_1\right) & s \leq \pi \\
 \int_0^{Q_1}\overline{cost_p}\left(s-\frac{P_{max} x}{E_{max}}\right)f_I(x)dx + \overline{cost_p}(\pi) \tilde{F}_I\left(Q_1 \right) & s> \pi
 \end{cases}$} \nonumber,
\end{eqnarray}
where
\begin{eqnarray}\label{eqn:15}
&&\overline{cost}_p(s^e)= \mathrm{E}_{E_{PFC}, q} \left[ cost_p\left( s^e\right)\right] \\
&=&  c_pp_1  \mathrm{E}_{E_{PFC}}\left[\left( E_{PFC}-\frac{E_{max}(1-s^e)}{\eta} \right)^+\right] \nonumber\\
&&+c_p p_{-1} \mathrm{E}_{E_{PFC}}\left[\left(E_{PFC}- \eta E_{max}s^e \right)^+\right] \nonumber
\end{eqnarray}
is the expected regulation failure penalty in the case that the SOC is $s^e$ when the frequency excursion occurs.

\section{Optimal BESS Control}\label{sec:optimalSoC}

In general, the optimal decision at each stage of a dynamic programming is a function of the system state observed at that stage. That is, we need to calculate the optimal charging target $\pi^*(s)$ as a function of the BESS SoC $s$ observed at the beginning of each $I$ interval. Interestingly, this is not necessary in our problem. The following theorem states that the optimal target SoC is a range that is \emph{invariant} with respect to the BESS SoC $s$ at each stage. Furthermore, the range converges to a single point $\pi^*$ that is \emph{independent} of $s$ when $\eta\rightarrow1$ or $c_e\ll c_p$,. This result is extremely appealing: we can pre-calculate $\pi^*$ for all stages offline. This greatly simplifies the system operation.

\begin{theorem}\label{thm:1}
The optimal target SoC that minimizes the cost $H^*(s)$ in \eqref{eqn:bellman} is a range $[\pi^*_{low}, \pi^*_{high}]$, where $\pi^*_{low}$ and $\pi^*_{high}$ are fixed in all stages regardless of the system state $s$. During each $I$ interval, the BESS is charged or discharged when its SoC falls outside the range, and remains idle when its SoC is in the range. In other words, at each stage, the optimal target SoC $\pi^*$ is set as
\begin{equation}
\pi^*:=\pi^*(s)=
\begin{cases}
\pi^*_{low} & \text{if}~s<\pi^*_{low} \\
\pi^*_{high} & \text{if}~s>\pi^*_{high} \\
s & \text{if}~s\in [\pi^*_{low}, \pi^*_{high}]
\end{cases}.
\end{equation}
Moreover, $\pi^*_{low}$ and $\pi^*_{high}$ converge to a single point $\pi^*$ when $\eta\rightarrow1$ or $c_e\rightarrow 0$.

\end{theorem}

To prove Theorem \ref{thm:1}, let us first characterise the sufficient and necessary conditions for optimal $\pi^*$. For convenience, rewrite \eqref{eqn:bellman} into
$$H^*(s)=\min_{\pi\in[0, 1]} H(s, \pi),$$
where
\begin{equation}\label{eqn:Bellman1}
H(s, \pi)=   h(s,\pi) +\alpha \mathrm{E}_{I, E_{PFC}, q}\left[H^*(g(s, \pi, I, E_{PFC}, q)\right].
\end{equation}
Taking the first order derivative $\frac{\partial H(s, \pi)}{\partial\pi}$, we obtain the following after some manipulations.
\begin{eqnarray}\label{eqn:derivative_pi_2}
&&\frac{\partial H(s, \pi)}{\partial\pi} =\frac{\partial h(s,\pi)}{\partial \pi} \\
& +& \resizebox{0.9\hsize}{!}{$\alpha p_1\tilde{F}_I\left(Q_1 \right)\int_0^\frac{(1-\pi)E_{max}}{\eta} \frac{\partial H^*\left(s\right)}{\partial s}\bigg|_{s=\pi+\frac{\eta e }{E_{max}}}f_{E_{PFC}}(e) de $}\nonumber\\
&+&\resizebox{0.9\hsize}{!}{$\alpha p_{-1} \tilde{F}_I\left(Q_1 \right) \int_0^{\eta \pi E_{max}} \frac{\partial H^*\left( s\right)}{\partial s}\bigg|_{s=\pi-\frac{e}{\eta E_{max}}} f_{E_{PFC}}(e)de $}.\nonumber
\end{eqnarray}
Specifically,
\begin{equation}\label{eqn:1stage_1deri}
\frac{\partial h(s, \pi)}{\partial \pi}=\frac{\partial}{\partial \pi} \mathrm{E}_I\left[cost_e(s, \pi)\right] + \frac{\partial}{\partial \pi} \mathrm{E}_{I} \left[ \overline{cost_p}\left( s^e(s, \pi)\right)\right],
\end{equation}
where
\begin{eqnarray}\label{eqn:17}
\frac{\partial}{\partial \pi}\mathrm{E}_{I} [ cost_e(s,\pi)] =
\begin{cases}
\frac{1}{\eta}c_e E_{max}  \tilde{F}_I\left( Q_1\right)  &\text{if}~ \pi > s \\
\eta c_e E_{max} \tilde{F}_I\left(Q_1\right) &\text{if}~ \pi < s
\end{cases}
\end{eqnarray}
as a result of differentiating \eqref{eqn:12},
and
\begin{eqnarray}\label{eqn:18}
&&\frac{\partial}{\partial \pi}\mathrm{E}_{I} [ \overline{cost_p}(s^e)] =\frac{\partial\overline{cost_p}(\pi)}{\partial\pi} \tilde{F}_I\left(Q_1\right)\\
&=&  \left(c_p p_1 \frac{E_{max}}{\eta } \tilde{F}_{E_{PFC}}\left(\frac{E_{max}(1-\pi)}{\eta }\right) \right. \nonumber\\
&& \left. - c_pp_{-1}\eta E_{max}\tilde{F}_{E_{PFC}}\left(\eta E_{max}\pi \right) \right)\tilde{F}_I\left(Q_1 \right) \nonumber
\end{eqnarray}
as a result of differentiating \eqref{eqn:14}\eqref{eqn:15}.
Note that $\mathrm{E}_{I} [ cost_e(s,\pi)]$ is not differentiable at $\pi=s$ unless $\frac{1}{\eta}c_e = \eta c_e$ (or equivalently when $\eta=1$ or $c_e=0$).

Substituting \eqref{eqn:17} and \eqref{eqn:18} to \eqref{eqn:derivative_pi_2}, we have
$$\frac{\partial H(s, \pi)}{\partial\pi} = \big(r(s,\pi)E_{max}+u(\pi)\big)\tilde{F}_I\left(Q_1 \right),$$
where $r(s, \pi)$ is defined in \eqref{eqn:r}  and
\begin{figure*}[t]
\begin{equation}\label{eqn:r}
r(s,\pi) =
\begin{cases}
r_1(\pi):=\frac{1}{\eta}c_e +  \frac{c_pp_1}{\eta } \tilde{F}_{E_{PFC}}\left(\frac{E_{max}(1-\pi)}{\eta }\right) - c_pp_{-1}\eta \tilde{F}_{E_{PFC}}\left(\eta E_{max}\pi \right) &\text{if}~ \pi > s\\
r_2(\pi):=\eta c_e +  \frac{c_pp_1}{\eta } \tilde{F}_{E_{PFC}}\left(\frac{E_{max}(1-\pi)}{\eta }\right) - c_pp_{-1}\eta \tilde{F}_{E_{PFC}}\left(\eta E_{max}\pi \right) &\text{if}~ \pi < s
\end{cases}.
\end{equation}
\end{figure*}

\begin{eqnarray}\label{eqn:22}
&&u(\pi)= \alpha p_1\int_0^\frac{(1-\pi)E_{max}}{\eta } \frac{\partial H^*\left(\pi+\frac{\eta e }{E_{max}}\right)}{\partial \pi} f_{E_{PFC}}(e) de\nonumber \\
&+&\alpha p_{-1}  \int_0^{\eta \pi E_{max}}   \frac{\partial H^*\left( \pi-\frac{e}{\eta E_{max}}\right)}{\partial \pi} f_{E_{PFC}}(e)de .
\end{eqnarray}

 To avoid trivial solutions, we assume that the CCDF $\tilde{F}_I\left(Q_1\right)>0$ for all $s, \pi$. Thus, the sign of $\frac{\partial}{\partial \pi} H(s, \pi) $ is determined by that of $r(s,\pi)E_{max}+u(\pi)$. As a result, the necessary condition for optimal $\pi^*$ is
\begin{equation}\label{eqn:necessary}
\resizebox{\hsize}{!}{$r(s,\pi^*)E_{max}+u(\pi^*)
\begin{cases}
=r_2(0)E_{max}+u(0)\geq 0 & \text{if}~\pi^*=0  \\
=r_2(\pi^*)E_{max}+u(\pi^*) =0 & \text{if}~\pi^*\in(0, s) \\
=r_1(\pi^*)E_{max}+u(\pi^*) =0 & \text{if}~\pi^*\in(s, 1) \\
=r_1(1)E_{max}+u(1) \leq 0 & \text{if}~\pi^*=1
\end{cases},$}
\end{equation}
when $\pi^*\neq s$. On the other hand, when $\pi^*=s$,
\begin{equation}\label{eqn:necessary2}
\resizebox{0.85\hsize}{!}{$r_1(s^+)E_{max}+u(s^+)>0~\text{and}~r_2(s^-)E_{max}+u(s^-)<0.$}
\end{equation}

Now we proceed to show that the necessary conditions \eqref{eqn:necessary} and   \eqref{eqn:necessary2} are also sufficient conditions for optimal $\pi^*$. To this end, let us first prove the convexity of $H^*(s)$ in the following proposition.

\begin{proposition}\label{pro:convexH}
$H^*(s)$ is convex in $s$. In other words, $\frac{\partial^2 H^*(s)}{\partial s^2}\geq 0$ for all $s$.
\end{proposition}
A key step to prove Proposition \ref{pro:convexH} is to show that $\frac{\partial^2 H^*(s)}{\partial s^2}$ is the fixed point of equation $f(s)=Tf(s)$, where operator $T$ is a contraction mapping. The details of the proof are deferred to Appendix \ref{append:A}.

Proposition \ref{pro:convexH} implies the following Lemma \ref{lem:4}, which further leads to Proposition \ref{pro:quasi}.
\begin{lemma}\label{lem:4}
Both $r_1(\pi)E_{max}+u(\pi)$ and $r_2(\pi)E_{max}+u(\pi)$ are increasing functions of $\pi$. Moreover, $r(s, \pi)E_{max}+u(\pi)$ is an increasing function of $\pi$.
\end{lemma}
The proof of the lemma is deferred to Appendix \ref{append:B}.

\begin{proposition}\label{pro:quasi}
$H(s, \pi)$ is a quasi-convex function of $\pi$. In other words, one of the following three conditions holds.
\begin{itemize}
\item[(a)] $\frac{\partial}{\partial \pi} H(s, \pi)\geq 0$ for all $\pi$.
\item[(b)] $\frac{\partial}{\partial \pi} H(s, \pi)\leq 0$ for all $\pi$.
\item[(c)] There exists a $\pi'$ such that $\frac{\partial}{\partial \pi} H(s, \pi)\leq 0$ when $\pi<\pi'$ and $\frac{\partial}{\partial \pi} H(s, \pi)\geq 0$ when $\pi>\pi'$.
\end{itemize}
\end{proposition}

%

The quasi-convexity of $H(s, \pi)$ is straightforward from Lemma \ref{lem:4}. It ensures that the necessary condition \eqref{eqn:necessary} and \eqref{eqn:necessary2} is also sufficient. We are now ready to prove our main result Theorem \ref{thm:1}.

\begin{proof}[Proof of Theorem \ref{thm:1}]
We calculate the optimal $\pi^*$ as follows. Let $\pi^*_{low}\in[0, 1]$ be the root of the equation
$$ r_1(\pi)E_{max}+u(\pi) =0.$$
In case the root does not exist\footnote{This happens when $r_1(0)E_{max}+u(0) >0$, i.e., $r_1(\pi)E_{max}+u(\pi) >0$ for all $\pi$, or when $ r_1(1)E_{max}+u(1)<0 $, i.e., $r_1(\pi)E_{max}+u(\pi) <0$ for all $\pi$.}, set $\pi^*_{low}=0$ if $r_1(0)E_{max}+u(0) >0$, and $\pi^*_{low}=1$ if $r_1(1)E_{max}+u(1) <0$. Similarly, define $\pi^*_{high}\in[0, 1]$ as the root of the equation
$$ r_2(\pi)E_{max}+u(\pi) =0.$$
In case the root does not exist, set $\pi^*_{high}=0$ if $r_2(0)E_{max}+u(0) >0$, and $\pi^*_{high}=1$ if $r_2(1)E_{max}+u(1) <0$.

From the definition,  $r_1(\pi)E_{max}+u(\pi)>r_2(\pi)E_{max}+u(\pi)$ for any given $\pi$. Thus, it always holds that $\pi^*_{low}\leq\pi^*_{high}$. From  the sufficient and necessary conditions in \eqref{eqn:necessary} and \eqref{eqn:necessary2}, we can conclude that
\begin{equation}
\pi^*=
\begin{cases}
\pi^*_{low} & \text{if}~s<\pi^*_{low} \\
\pi^*_{high} & \text{if}~s>\pi^*_{high} \\
s & \text{if}~s\in [\pi^*_{low}, \pi^*_{high}]
\end{cases}.
\end{equation}
In other words, the optimal target SoC is a range $[\pi^*_{low},\pi^*_{high}]$. Since $r_1(\pi)E_{max}+u(\pi)$ and $r_2(\pi)E_{max}+u(\pi)$ are not functions of $s$, $\pi^*_{low}$ and $\pi^*_{high}$ are independent of $s$. Thus, the range $[\pi^*_{low},\pi^*_{high}]$ is fixed for all stages regardless of the system state $s$.

Furthermore, when $\eta= 1$ or $c_e= 0$,  $r_1(\pi)=r_2(\pi)$ for all $\pi$. In this case, $\pi^*_{low}=\pi^*_{high}$. Thus, the optimal $\pi^*$ becomes a single point that remains constant for all system states $s$. This completes the proof.
\end{proof}

\begin{remark}
Usually, infinite-horizon dynamic programming problems are solved by value iteration or policy iteration methods \cite{bertsekas1995dynamic}.  Therein, an $N$-dimensional decision vector is optimized in each iteration, with each entry of the vector being the optimal decision corresponding to a system state. In our problem, the system state $s$ is continuous in $[0, 1]$. Discretizing it can lead to a large $N$. Fortunately, the results in this section show that the optimal decision is characterized by two scalars $\pi^*_{low}$ and $\pi^*_{high}$ that remain constant for all system states. Thus, the calculation of the optimal decision is greatly simplified. A brief discussion on the algorithm to obtain  $\pi^*_{low}$ and $\pi^*_{high}$ can be found in Appendix \ref{append:algorithm}.
\end{remark}

\section{Optimal BESS Planning}\label{sec:optimalPlan}
Obviously, the minimum operating cost $H^*(s)$ is a function of the BESS energy capacity $E_{max}$. On the other hand, the capital cost of acquiring and setting up the BESS increases with $E_{max}$. Let the capital cost be denoted as $Q(E_{max})$, which is an increasing function of $E_{max}$. In this section, we are interested in investigating the optimal $E_{max}$ that minimizes the total expected cost
$\lambda Q(E_{max}) + \mathrm{E}_s\left[H^*(s)\right],$
where $\lambda$ is a weighting factor that depends on the BESS life time, BESS degradation, and the tendering period.  $\mathrm{E}_s\left[H^*(s)\right]$ is the expected value of $H^*(s)$ over all initial SoC $s$ under the optimal charging operation.

The main result of this section is given in Theorem \ref{thm:capacity} below, which states that $H^*(s)$ is a decreasing convex function of $E_{max}$ for all $s$. As a result, $ \mathrm{E}_s\left[H^*(s)\right]$ is also a decreasing convex function of $E_{max}$. In other words, the marginal decrease of the $\mathrm{E}_s\left[H^*(s)\right]$ diminishes when $E_{max}$ becomes large. This implies the existence of a unique optimal $E_{max}$, at which the marginal increase of $Q(E_{max})$ is equal to the marginal decrease of $\mathrm{E}_s\left[H^*(s)\right]$, i.e.,
$$\lambda \frac{\partial Q(E_{max})}{\partial E_{max}} = - \frac{\partial \mathrm{E}_s\left[H^*(s)\right]}{\partial E_{max}}.$$

\begin{theorem}\label{thm:capacity}
The minimum operating cost $H^*(s)$ given in \eqref{eqn:bellman} is a decreasing convex function of $E_{max}$.
\end{theorem}
The proof of Theorem \ref{thm:capacity} is deferred to Appendix \ref{append:C}

\section{Numerical Results}\label{sec:numerical}

In this section, we validate our analysis and investigate how different system parameters affect the optimal BESS operation and planning. The simulations are conducted using the real-time frequency measurement data collectd in Sacramanto, CA, as shown in Fig. \ref{fig:data}. The sample rate is 10 Hz (i.e., 1 measurement per 0.1 seconds). The data set, provided by FNET/GridEye \cite{7265090}, includes a total of 2,555,377 samples, accounting for about 71 hours of frequency measurement. Suppose that a frequency excursion event occurs when the system frequency deviates outside a dead band of 10mHz around the normative frequency. The empirical distributions of $I$, $J$, and $q$ derived from the measurement data are plotted in Fig. \ref{fig:CDF}.

\begin{figure}
\centering
\includegraphics[width=0.5\textwidth]{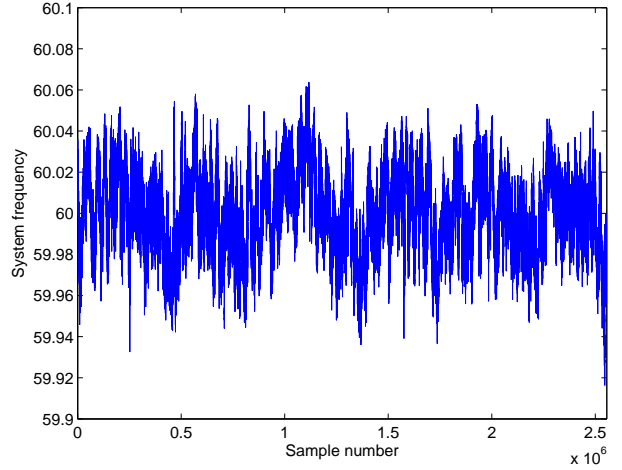}
\caption{System frequency measured at  Sacramanto, CA}\label{fig:data}
\end{figure}

An underlying assumption of our analysis is that $I_n$, $J_n$ and $q_n$ are i.i.d. for different $n$, respectively, and that they are mutually independent. To validate this assumption, we plot the auto-correlations and cross-correlations of the variables in Figs. \ref{fig:auto} and \ref{fig:cross}, respectively. As we can see from Fig. \ref{fig:auto}, the auto-correlations of the variables reach the peak when the time lag is 0 and are close to zero at non-zero time lags, implying that they are approximately independent for different $n$. Likewise, Fig. \ref{fig:cross} shows that the cross-correlations of the variables are all close to zero, implying that $I$, $J$, and $q$ are mutually independent.

\begin{figure}
\centering
\includegraphics[width=0.5\textwidth]{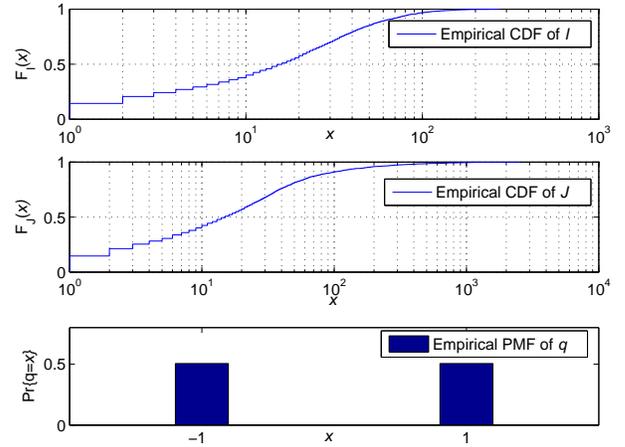}
\caption{Empirical distributions of $I$, $J$, and $q$}\label{fig:CDF}
\end{figure}

\begin{figure}
\centering
\includegraphics[width=0.5\textwidth]{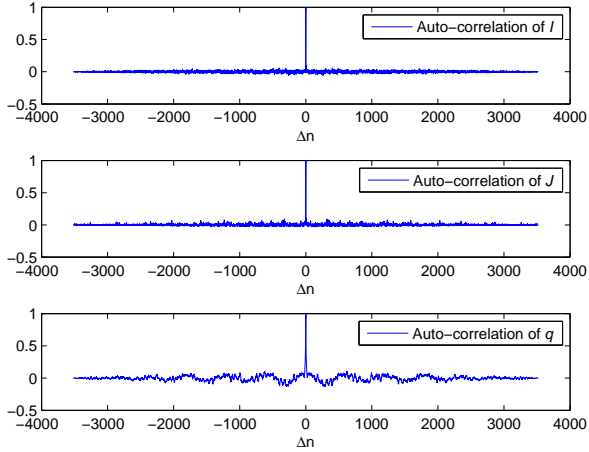}
\caption{Auto-correlations of $I$, $J$, and $q$}\label{fig:auto}
\end{figure}

\begin{figure}
\centering
\includegraphics[width=0.5\textwidth]{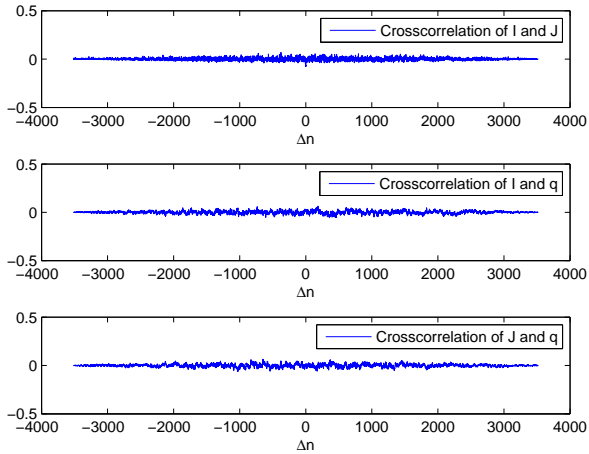}
\caption{Cross-correlations of the variables.}\label{fig:cross}
\end{figure}

Before proceeding, let us verify Proposition \ref{pro:convexH}, the convexity of $H^*(s)$ with respect to $s$, which is a key step in the proof of our main result. Unless otherwise stated, we assume that $E_{max}=0.1$MWh, $P_{max}= 1$MW, $P_{PFC}$ is uniformly distributed in $[0.5, 1]$MW, and the discount factor $\alpha=0.9$ in the rest of the section. In Fig. \ref{fig:H_vs_s}, we plot $H^*(s)$ against $s$  when $c_e=\$0.1/$kWh and $c_p=\$10/$kWh. The figure verifies that $H^*(s)$ is indeed a convex function of $s$, as proved in  Proposition \ref{pro:convexH}.

 \begin{figure}
\centering
\includegraphics[width=0.5\textwidth]{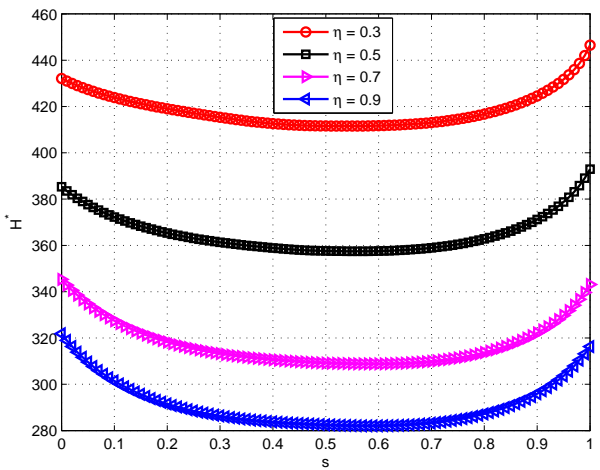}
\caption{Convexity of $H^*(s)$ with respect to $s$ when $c_e=\$0.1/$kWh and $c_p=\$10/$kWh.}\label{fig:H_vs_s}
\end{figure}

\subsection{Optimal Target SoC}

In this subsection, we investigate the effect of various system parameters on the optimal target SoC $\pi_{low}^*$ and $\pi_{high}^*$. The  settings of system parameters are the same as that in Fig. \ref{fig:H_vs_s} unless otherwise stated. In Fig. \ref{fig:pi_vs_eta}, $\pi_{low}^*$ and $\pi_{high}^*$ are plotted against $\eta$. It can be seen that when the battery efficiency $\eta$ is  low, $[\pi^*_{low}, \pi^*_{high}]$ is a relatively wide interval. The interval narrows when $\eta$ becomes large, and converges to a single point when $\eta\rightarrow 1$. This is consistent with Theorem \ref{thm:1}. Recall that there is no need to charge or discharge the battery during an $I$ interval if the SoC at the beginning of the $I$ interval is already within $[\pi^*_{low}, \pi^*_{high}]$. The result in  Fig. \ref{fig:pi_vs_eta} is intuitive in the sense that when the battery efficiency is low, adjusting SoC during the $I$ intervals is more costly due to  power losses. Thus, the interval $[\pi^*_{low}, \pi^*_{high}]$ is wider so that the battery SoC does not need to be adjusted too often.

In Fig. \ref{fig:pi_vs_cp}, $\pi^*_{low}$ and $\pi^*_{high}$ are plotted against $c_p$ when $c_e=\$0.1/$kWh and BESS efficiency $\eta=0.8$. The figure shows that $[\pi^*_{low}, \pi^*_{high}]$ is a relatively large interval when $c_p$ is comparable with $c_e$. When $c_p$ becomes large compared with $c_e$, $\pi^*_{low}$ and $\pi^*_{high}$ converges to a single point, as proved in Theorem \ref{thm:1}. Indeed,
$\pi^*_{low}$ and $\pi^*_{high}$ overlap when $c_p$ is larger than  $\$35$/kWh. In practice, the regulation failure penalty $c_p$ is usually much larger the regular electricity price $c_e$. Thus, we can safely regard the optimal target SoC as a single point in practical system designs.

Fig. \ref{fig:pi_vs_E} investigates the effect of battery energy capacity $E_{max}$ on the optimal target SoC $\pi_{low}^*$ and $\pi_{high}^*$. It can be seen that both $\pi_{low}^*$ and $\pi_{high}^*$ become low when $E_{max}$ is very large. This can be intuitively explained as follows. Recall that $s^e$ is to denote the BESS SoC at the end of an $I$ interval (or the beginning of a $J$ interval). If $s^eE_{max}$ and $(1-s^e)E_{max}$ are both larger than the maximum possible $E_{PFC}$, then regulation failures are completely avoided, and the operating cost would be dominated by the charging cost during $I$ intervals. When $E_{max}$ is large, there is a wide range of $s^e$ that can completely prevent regulation failures. Out of this range, smaller $s^e$'s are preferred, so that the charging cost during $I$ intervals is lower. This, the optimal target SoCs must be low when $E_{max}$ becomes large.

\begin{figure}
\centering
\includegraphics[width=0.5\textwidth]{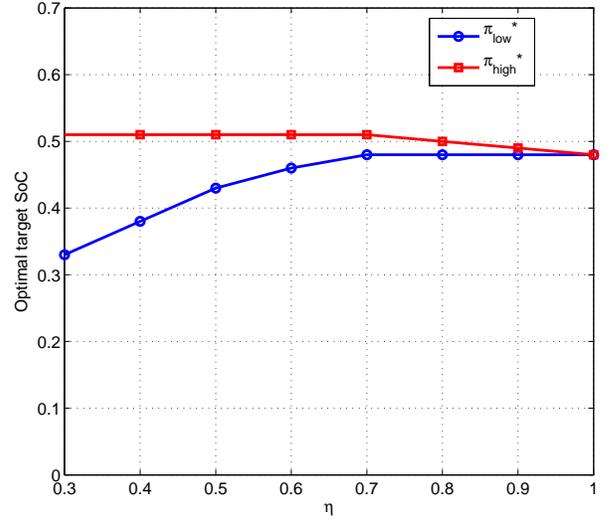}
\caption{$\pi^*_{low}$ and $\pi^*_{high}$ versus $\eta$ when $c_e=\$0.1/$kWh and $c_p=\$10/$kWh.}\label{fig:pi_vs_eta}
\end{figure}

\begin{figure}
\centering
\includegraphics[width=0.5\textwidth]{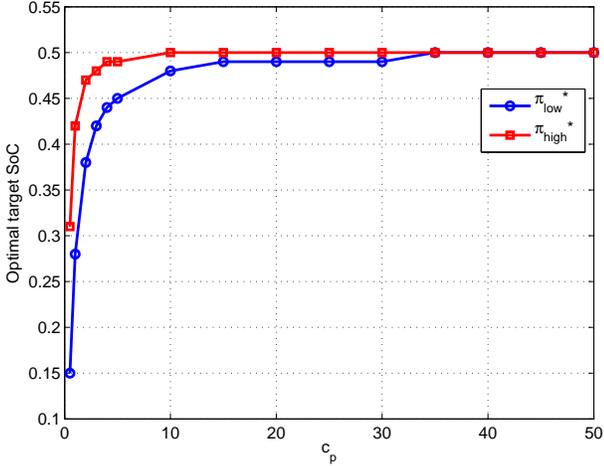}
\caption{$[\pi^*_{low}, \pi^*_{high}]$ vs. $c_p$ when $c_e=\$0.1/$kWh and $\eta$=0.8}\label{fig:pi_vs_cp}
\end{figure}

\begin{figure}
\centering
\includegraphics[width=0.5\textwidth]{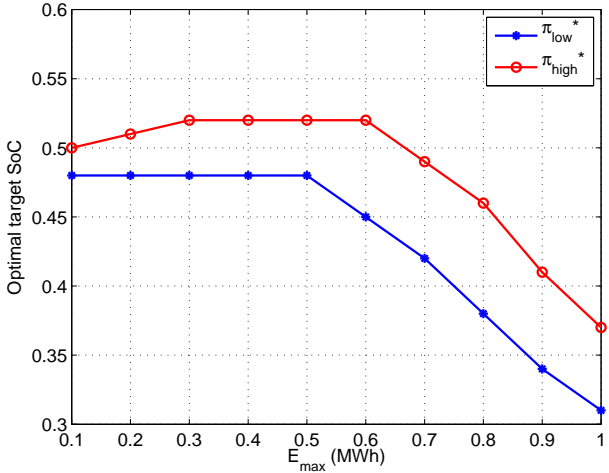}
\caption{$[\pi^*_{low}, \pi^*_{high}]$ vs. $E_{max}$ when $c_e=\$0.1/$kWh and $c_p=\$10/$kWh.}\label{fig:pi_vs_E}
\end{figure}

\subsection{Time Response Comparison}

To illustrate the advantage of the proposed BESS control scheme, we compare the operating cost of our scheme with the following three benchmark algorithms proposed in previous work, e.g., in \cite{oudalov2007optimizing}.
\begin{itemize}
\item No additional charging during $I$ intervals. Referred to ``No recharging" in the figures.
\item Recharge up to $100\%$ during $I$ intervals. Referred to as ``Aggressive recharging" in the figures.
\item Recharge with upper and lower target SoCs. This scheme is similar to our proposed scheme, except that the target SoCs are set heuristically (instead of optimized in our algorithm). In upper and lower target SoCs are set to be $0.92$ and $0.73$, respectively in \cite{oudalov2007optimizing}. This scheme is referred to as ``Heuristic recharging" in the figures.
\end{itemize}
In particular, we run a time-response simulation using the real-time frequency measurement data in Fig. \ref{fig:data}. The probability of encountering regulation failures is plotted in Fig. \ref{fig:outage}. Moreover, the time-aggregate operating costs (without discounting) are plotted in Figs. \ref{fig:timeres1} and \ref{fig:timeres2} when $E_{max}=0.1$MWh and  $E_{max}=1.5$MWh, respectively. It can be seen from Figs. \ref{fig:timeres1} and \ref{fig:timeres2}  that both "No recharging" and "Aggressive recharging" algorithms yield much higher cost than the optimal algorithm proposed in the paper. This is because the battery SoC is often too low (with "No recharging") or too high (with "Aggressive charging"), yielding much higher regulation failure probabilities, as shown in Fig. \ref{fig:outage}. On the other hand, with optimal target SoC, the proposed algorithm reduces both the operating cost and regulation failure probability compared with "Heuristic recharging".

\begin{figure}
\centering
\includegraphics[width=0.5\textwidth]{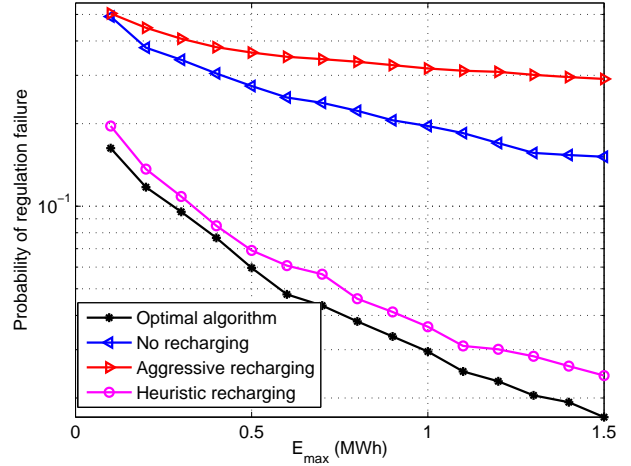}
\caption{Comparison of regulation failure probability when $c_e=\$0.1/$kWh, $c_p=\$10/$kWh and $\eta$=0.8}\label{fig:outage}
\end{figure}

\begin{figure}
\centering
\includegraphics[width=0.5\textwidth]{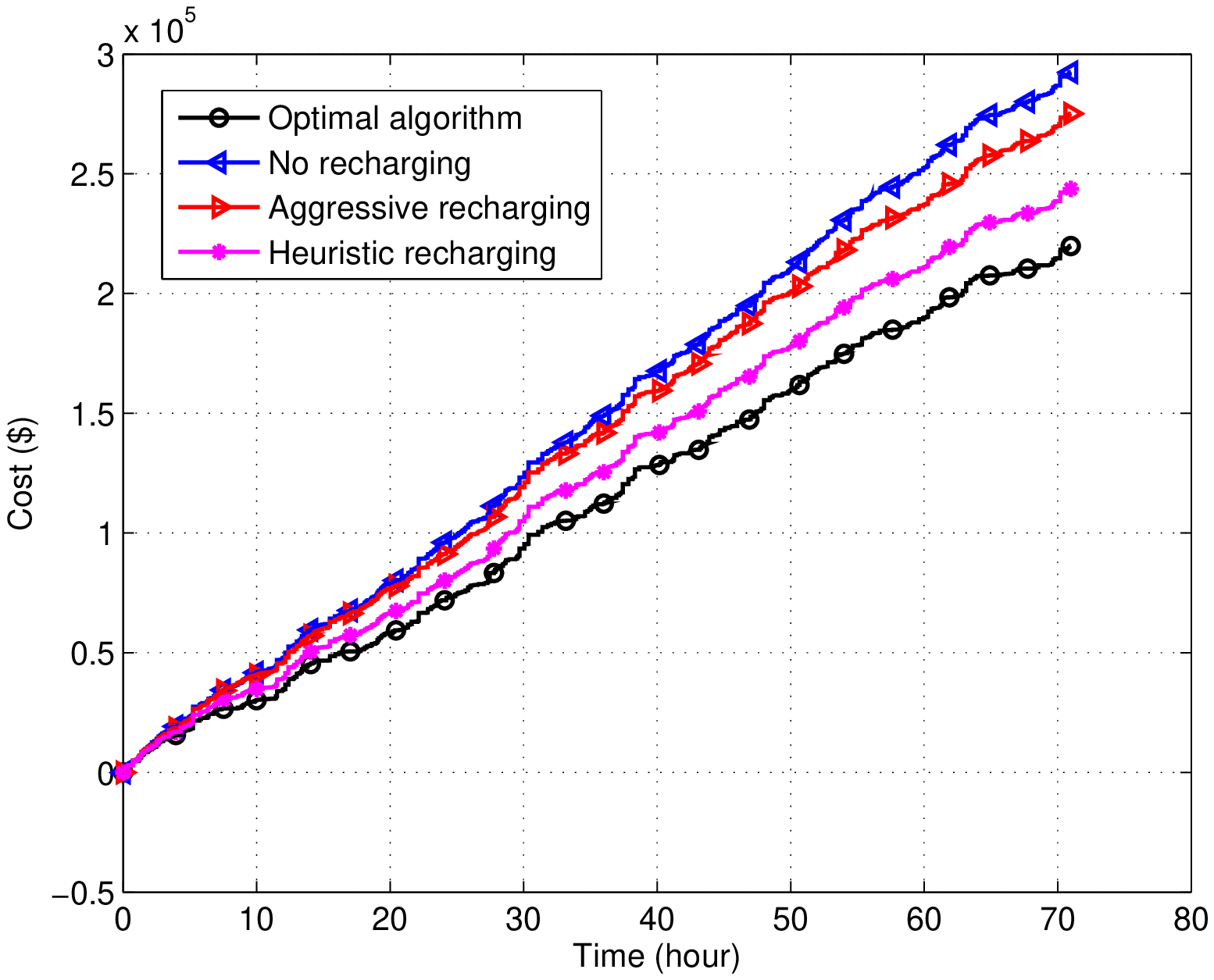}
\caption{Comparison of time-aggregate costs when $E_{max}=0.1$MWh, $c_e=\$0.1/$kWh, $c_p=\$10/$kWh and $\eta$=0.8}\label{fig:timeres1}
\end{figure}

\begin{figure}
\centering
\includegraphics[width=0.5\textwidth]{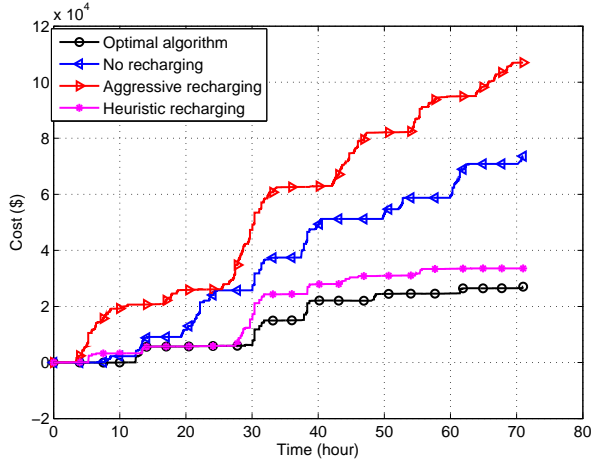}
\caption{Comparison of time-aggregate costs when $E_{max}=1.5$MWh, $c_e=\$0.1/$kWh, $c_p=\$10/$kWh and $\eta$=0.8}\label{fig:timeres2}
\end{figure}

\subsection{Optimal BESS Planning}

In Fig. \ref{fig:H_vs_E}, we verify Theorem \ref{thm:capacity} and investigate the effect of BESS energy capacity $E_{max}$ on the operating cost $H^*$. Here, $c_e=\$0.1/$kWh, $c_p=\$10/$kWh,  $\eta=0.8$, and $E_{max}$ varies from 0.05MWh to to 10MWh. It can be see that $H^*(s)$ is a decreasing convex function of $E_{max}$ for all initial SoC $s$. This implies that there exists an optimal BESS energy capacity $E_{max}$ that hits the optimal balance between the capital investment and operating cost.

\begin{figure}
\centering
\includegraphics[width=0.5\textwidth]{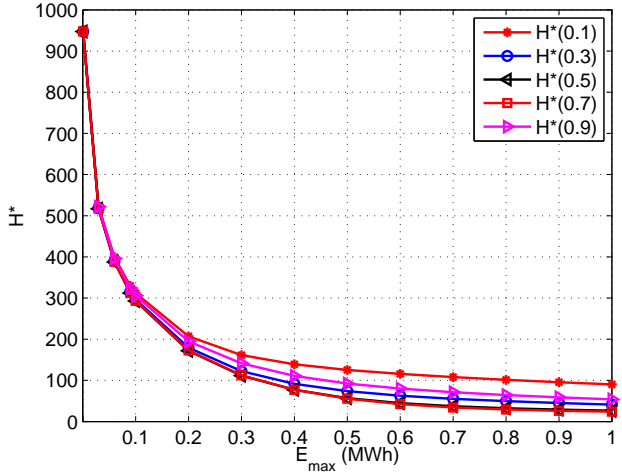}
\caption{$H^*(s)$ vs. $E_{max}$ when $c_e=\$0.1/$kWh, $c_p=\$10/$kWh, and $\eta=0.8$.}\label{fig:H_vs_E}
\end{figure}

\section{Conclusions}\label{sec:conclusions}
We studied the optimal planning and control for BESSs participating in the PFC regulation market. We show that the optimal BESS control is to charge or discharge the BESS during $I$ intervals until its SoC reaches a target value. We have proved that the optimal target SoC is a range that is invariant with respect to the BESS SoC $s$ at the beginning of the $I$ intervals. This implies that the optimal target SoC can be calculated offline and remain unchanged over the entire system time. Hence, the operation complexity can be kept very low. Moreover, the target SoC range reduces to a point in practical systems, where the penalty rate for regulation failure is much larger than the regular electricity price. It was also shown that the optimal operating cost is a decreasing convex function of the BESS energy capacity, implying the existence of an optimal energy capacity that balances the capital investment of BESS and the operating cost.

Other than PFC, BESSs can serve multiple purposes, such as demand response, energy arbitrage, and peak shaving. Different services require different energy and power capacities. For example, PFC reserves do not require high energy capacity, but are sensitive to regulation failures. On the other hand, high energy capacity is needed for demand response, energy arbitrage, and peak shaving. It is an interesting future research topic to study the optimal combining of these services in a single BESS.

\begin{appendices}
\section{Proof of Proposition \ref{pro:convexH}}\label{append:A}
\begin{proof}
First, calculate
\begin{eqnarray}
&&\frac{\partial H^*(s)}{\partial s} = \frac{\partial h(s, \pi^*)}{\partial s} \\
& +& \alpha p_1 \int_0^{Q_1^*} \int_0^{Q_2^*} \frac{\partial H^*(Q_3^*)}{\partial s} f_{E_{PFC}}(e)f_I(i) dedi \nonumber\\
&+&\alpha p_{-1}\int_0^{Q_1^*} \int_0^{Q_4^*}  \frac{\partial H^*(Q_5^*)}{\partial s}f_{E_{PFC}}(e)f_I(i) dedi \nonumber
\end{eqnarray}
where
\begin{equation}
Q_2^*=\frac{(1-s)E_{max}-\mathrm{sgn}(\pi^*-s)P_{max} i}{\eta },
\end{equation}
\begin{equation}
Q_3^*= \frac{sE_{max}+\mathrm{sgn}(\pi^*-s)P_{max} i+\eta e}{E_{max}},
\end{equation}
\begin{equation}
Q_4^*=\eta(sE_{max}+\mathrm{sgn}(\pi^*-s)P_{max} i),
\end{equation}
\begin{equation}
Q_5^*=\frac{sE_{max}+\mathrm{sgn}(\pi^*-s)P_{max} i- e/\eta}{E_{max}} ,
\end{equation}
and $Q_1^*$ is the same as $Q_1$ in \eqref{eqn:q1} except that $\pi$ is replaced by $\pi^*$ in the definition.
After some manipulations, we have
\begin{eqnarray} \label{eqn:2nd_3}
&&H^{*(2)}(s) =a(s) \\
&+& \alpha p_1 \int_0^{Q_1^*} \int_0^{Q_2^*} H^{*(2)}\left(s'\right)\big|_{s'=Q_3^*} f_{E_{PFC}}(e)f_I(i) dedi \nonumber\\
&+&\alpha p_{-1}\int_0^{Q_1^* }\int_0^{Q_4^*}  H^{*(2)}\left(s'\right)\big|_{s'=Q_5^*} f_{E_{PFC}}(e)f_I(i) dedi, \nonumber
\end{eqnarray}
where $H^{*(2)}(s):= \frac{\partial^2 H^*(s)}{\partial s^2}$ and
\begin{eqnarray}\label{eqn:a}
a(s) &=&\frac{-\mathrm{sgn}(\pi^*-s)E_{max}}{P_{max}}\big(r(s,\pi^*)E_{max}+u(\pi^*)\big)f_I\left(Q_1^*\right) \nonumber\\
&&+\int_0^{Q_1^*}\frac{\partial^2 \overline{cost}_p\left(s'\right)}{\partial s'^2}\bigg|_{s'=s+\frac{P_{max} i}{E_{max}}}f_I(i)di.
\end{eqnarray}
We claim that $a(s)$ is non-negative for all $s$. To see this, note that
$$-\mathrm{sgn}(\pi^*-s)\big(r(s,\pi^*)E_{max}+u(\pi^*)\big)\geq 0$$
for all $s$ due to the necessary condition of optimal $\pi^*$ in \eqref{eqn:necessary} and \eqref{eqn:necessary2}. Thus, the first term of $a(s)$ is non-negative. Moreover, the integrand in the second term of $a(s)$ is always non-negative as:
\begin{equation}
\resizebox{\hsize}{!}{$
\begin{aligned}
&\frac{\partial^2 \overline{cost_p}(s)}{\partial s^2} \\
&=c_p E_{max}^2 \left(\frac{p_1}{\eta^2} f_{E_{PFC}}\left(E_{max}(1-s)\right) +p_{-1} \eta^2 f_{E_{PFC}}\left(E_{max}s \right) \right)\\
&\geq 0,
\end{aligned}$}
\end{equation}
where the equality is obtained by taking the second-order derivative of \eqref{eqn:15} over $s^e$ at $s^e = s$, and the inequality is due to the fact that PDF functions are non-negative. Thus, $a(s)\geq 0$.

Define two operators $D$ and $T$ such that
\begin{eqnarray}
Df(s) &=&  \alpha p_1 \int_0^{Q_1^*} \int_0^{Q_2^*} f\left(s'\right)\big|_{s'=Q_3^*} f_{E_{PFC}}(e)f_I(i) dedi \nonumber\\
&+&\alpha p_{-1}\int_0^{Q_1^* }\int_0^{Q_4^*}  f\left(s'\right)\big|_{s'=Q_5^*} f_{E_{PFC}}(e)f_I(i) dedi, \nonumber
\end{eqnarray}
and
\begin{equation}\label{eqn:T}
Tf(s)=a(s)+Df(s).
\end{equation}
It will be shown in Lemma \ref{lem:contraction} that the operator $T$ is a contraction mapping. Thus, $H^{*(2)}(s)$ is the fixed point of equation $f(s)=Tf(s),$
and the fixed point can be achieved by iteration
$$f^{(k+1)}(s) = Tf^{(k)}(s).$$
Letting $f^{(0)}(s)=0$ for all $s$, we can calculate the fixed point as
$$H^{*(2)}(s) = \sum_{i=0}^\infty K_i(s),$$
where $K_0(s)=a(s)$ and $K_i(s)=DK_{i-1}(s)$. Note that $D$ is a summation of two integrals, and therefore is non-negative when the integrand is non-negative. Thus, all $K_i(s)\geq 0$, because $K_0(s)=a(s)\geq 0$. As a result, $H^{*(2)}(s)\geq 0$ for all $s$. This completes the proof.
\end{proof}

\begin{lemma}\label{lem:contraction}
The operator $T$ defined in \eqref{eqn:T} is a contraction mapping.
\end{lemma}

To prove the lemma, we can show that $T$ satisfies following Blackwell Sufficient Conditions for contraction mapping.
\begin{itemize}
\item (Monotonicity) For any pairs of functions $f(s)$ and $g(s)$ such that $f(s)\leq g(s)$ for all $s$, $Tf(s)\leq Tg(s)$.
\item (Discounting) $\exists \beta \in(0, 1): T(f+b)(s)<Tf(s)+\beta b~ \forall f, b\geq0, s$.
\end{itemize}
\begin{proof}
Obviously, $Df(s)\leq Dg(s)$ for any pairs of functions $f(s)\leq g(s)$, because the operators is a summation of two integrals with non-negative integrands. Thus, $Tf(s)\leq Tg(s)$, and the Monotonicity condition holds.

To prove the discounting property, notice that
\begin{eqnarray}
T(f+b)(s)&=&a(s)+D(f+b)(s)=a(s)+Df(s)+Db \nonumber\\
&=&Tf(s)+Db,
\end{eqnarray}
because integrals are linear operations. Moreover,
\begin{eqnarray}
Db& =&  \alpha b \left( p_1 \int_0^{Q_1^*} \int_0^{Q_2^*}  f_{E_{PFC}}(e)f_I(i) dedi \right.\nonumber\\
&&\left.+p_{-1}\int_0^{Q_1^* }\int_0^{Q_4^*}   f_{E_{PFC}}(e)f_I(i) dedi\right)\nonumber\\
&\leq& \alpha b \left( p_1 \int_0^{Q_1^*}f_I(i) di +p_{-1}\int_0^{Q_1^*} f_I(i) di\right)\nonumber\\
&\leq& \alpha b(p_1+p_{-1}) \nonumber\\
&=&\alpha b.
\end{eqnarray}
Here, the inequalities are due to the fact that the integrals of PDF functions are no larger than 1. Since $\alpha$ is a discounting factor that is smaller than 1, the Discounting condition holds.
\end{proof}

\section{Proof of Lemma \ref{lem:4}}\label{append:B}
\begin{proof}
\begin{eqnarray}
&&\frac{\partial u(\pi)}{\partial \pi} \nonumber\\
&=&\alpha p_1\int_0^\frac{(1-\pi)E_{max}}{\eta }  H^{*(2)}\left(s\right)\big|_{s=\pi+\frac{\eta e }{E_{max}}} f_{E_{PFC}}(e) de \nonumber\\
&&+\alpha p_{-1}  \int_0^{\eta \pi E_{max}}   H^{*(2)}\left( s\right)\big|_{s=\pi-\frac{e}{\eta E_{max}}} f_{E_{PFC}}(e)de \nonumber\\
&\geq& 0,
\end{eqnarray}
where the equality is obtained by differentiating \eqref{eqn:22} over $\pi$, and the inequality is due to the fact that $H^{*(2)}\left(s\right) \geq 0$ for all $s$, as proved in Proposition \ref{pro:convexH}. Thus, $u(\pi)$ increases with $\pi$. Meanwhile, both $r_1(\pi)$ and $r_2(\pi)$  are increasing functions of $\pi$, because $\tilde{F}_{E_{PFC}}(x)$ is a decreasing function of $x$. Hence, both $r_1(\pi)E_{max}+u(\pi)$ and $r_2(\pi)E_{max}+u(\pi)$ are increasing functions of $\pi$. Moreover, when $\pi$ increases from $s^-$ to $s^+$, $r(s,\pi)E_{max}+u(\pi)$ increases by $\left(\frac{1}{\eta}-\eta\right)c_e$ from $r_2(s^-)E_{max}+u(s^-)$ to $r_1(s^+)E_{max}+u(s^+)$. This completes the proof.
\end{proof}

\section{Proof of Theorem \ref{thm:capacity}}\label{append:C}
\begin{proof}
The proof of convexity of $H^*(s)$ with respect to $E_{max}$ is similar to that for Proposition \ref{pro:convexH}, and thus is shortened here. We first calculate
\begin{equation}\label{eqn:second_E}
\resizebox{\hsize}{!}{$
\begin{aligned}
&\frac{\partial^2 H^*(s)}{\partial E_{max}^2} \\
=& \tilde{a}(s, E_{max}) \\
&+ \alpha p_1 \int_0^{Q_1^*} \int_0^{Q_2^*}  \frac{\partial^2 H(s')}{\partial E_{max}^2}\bigg|_{s' = Q_3^*}  f_{E_{PFC}}(e) f_I(i) de di \\
&+ \alpha p_1 \tilde{F}_I(Q_1^*) \int_0^{\frac{(1-\pi^*)E_{max}}{\eta}}  \frac{\partial^2 H(s')}{\partial E_{max}^2}\bigg|_{s' = \frac{\pi^*E_{max}+\eta e}{E_{max}}} f_{E_{PFC}}(e) de \\
&+ \alpha p_{-1} \int_0^{Q_1^*} \int_0^{Q_4^*}  \frac{\partial^2 H(s')}{\partial E_{max}^2}\bigg|_{s' =Q_5^*} f_{E_{PFC}}(e) f_I(i) de di \\
&+ \alpha p_{-1} \tilde{F}_I(Q_1^*) \int_0^{\eta\pi^*E_{max}}  \frac{\partial^2 H(s')}{\partial E_{max}^2}\bigg|_{s' = \pi^*-\frac{e}{\eta E_{max}}} f_{E_{PFC}}(e) de, \\
\end{aligned}$}
\end{equation}
where

\begin{equation}
\resizebox{\hsize}{!}{$
\begin{aligned}
&\tilde{a}(s, E_{max}) \\
=& \frac{-\mathrm{sgn}(\pi^*-s)|\pi^*-s|^2}{P_{max} E_{max}} f_I (Q_1^*)  \left(r(s, \pi^*)E_{max} + u(\pi^*)\right)\\
&+c_p p_1 \left (\frac{1-\pi^*}{\eta} \right)^2 f_{E_{PFC}}\left(\frac{(1-\pi^*)E_{max}}{\eta}\right)\tilde{F}_I (Q_1^*)\\
& + c_p p_{-1} (\eta \pi^*)^2 f_{E_{PFC}} \left(\eta\pi^*E_{max}\right)  \tilde{F}_I (Q_1^*)  \\
&+ \int_0^{Q_1^*} c_p \left( p_1 \left (\frac{1-s}{\eta}\right)^2 f_{E_{PFC}}(Q_2^*) + p_{-1} (\eta s)^2 f_{E_{PFC}} (Q_4^*) \right) f_I(i)di.
\end{aligned}$}
\end{equation}
We claim that $\tilde{a}(s, E_{max})\geq 0$ for all $s$ and $E_{max}$. To see this, note that the first term is always non-negative, because
$$-\mathrm{sgn}(\pi^*-s)\left(r(s, \pi^*)E_{max} + u(\pi^*)\right)\geq 0$$
due to \eqref{eqn:necessary} and \eqref{eqn:necessary2}. Moreover, the remaining terms are non-negative due to the non-negativeness of PDFs and CCDFs.

Same as the proof in Proposition \ref{pro:convexH}, we can show that the right hand side of \eqref{eqn:second_E} is a contraction mapping. Thus, we can calculate $\frac{\partial^2 H^*(s)}{\partial E_{max}^2}$ as a fixed point and get
$$\frac{\partial^2 H^*(s)}{\partial E_{max}^2}=\sum_{i=0}^\infty \tilde{K}_i(s, E_{max}),$$
where all $\tilde{K}_i(s, E_{max})\geq 0$. This implies that $\frac{\partial^2 H^*(s)}{\partial E_{max}^2}\geq0$, and thus $H^*(s)$ is convex with respect to $E_{max}$.

Now we proceed to prove that $H^*(s)$ is a decreasing function of $E_{max}$. We first show that the optimal single-stage cost $h^*(s)=\min_{\pi}h(s, \pi)$ decreases with $E_{max}$. Then, the decreasing monotonicity of  $H^*(s)$ with respect to $E_{max}$ can be proved by the monotonicity property of contraction mapping, which is stated in Lemma \ref{lem:contraction}.

Recall that $\frac{\partial h(s, \pi)}{\partial \pi}=r(s, \pi)E_{max}\tilde{F}_I(Q_1)$,  where $r(s, \pi)$ is defined in \eqref{eqn:r}.  Thus, the optimal $\pi$ that minimizes $h(s, \pi)$ satisfies
\begin{equation} \label{eqn:r=0}
r(s, \pi)=0.
\end{equation}
Furthermore, we can calculate that
\begin{eqnarray} \label{eqn:h_wrt_E}
&&\frac{\partial h(s, \pi)}{\partial E_{max}} = \left(\mathbf{1}_{\pi\geq s}\frac{1}{\eta} - \mathbf{1}_{\pi<s} \eta \right) c_e |\pi-s| \tilde{F}_I\left( Q_1\right)\nonumber \\
&-& c_p\frac{p_1(1-\pi)}{\eta}\tilde{F}_{E_{PFC}}\left(\frac{E_{max}(1-\pi)}{\eta} \right)\tilde{F}_I\left( Q_1\right) \nonumber\\
& -& c_p p_{-1}\eta\pi \tilde{F}_{E_{PFC}}\left(\eta E_{max}\pi \right) \tilde{F}_I\left( Q_1\right)
\end{eqnarray}
Substituting \eqref{eqn:r=0} to \eqref{eqn:h_wrt_E}, we have
\begin{eqnarray} \label{eqn:h_wrt_E_2}
&&\frac{\partial h^*(s)}{\partial E_{max}} = - \left( \mathbf{1}_{\pi_{1st}\geq s}\frac{1}{\eta} + \mathbf{1}_{\pi_{1st}<s} \eta \right) c_e s \tilde{F}_I\left( Q_1\right)\nonumber\\
& -&\frac{c_pp_1}{\eta} \tilde{F}_{E_{PFC}}\left(\frac{E_{max}(1-\pi_{1st})}{\eta} \right)\tilde{F}_I\left( Q_1\right) \nonumber\\
&\leq& 0,
\end{eqnarray}
where $\pi_{1st}$ is the minimizer of $h(s, \pi)$. \eqref{eqn:h_wrt_E_2} implies that $h^*(s)$ decreases with $E_{max}$ for all $s$.

Next, note that the Bellman equation of infinite-horizon dynamic programming is a contraction mapping \cite{bertsekas1995dynamic}. Let
\begin{equation}
TH(s)=\min_{\pi\in[0,1]} h(s, \pi) + \alpha \mathrm{E}_{I, E_{PFC}, q}\left[ H(g(s, \pi, I, E_{PFC}, q))\right]
\end{equation}
be the contraction operator  corresponding to the Bellman equation in \eqref{eqn:bellman}.  Then,
$$H^*(s) = \lim_{k\rightarrow\infty}(T^k H_0)(s)$$
for all $s$.

Starting with $H_0(s)=0$, we have
$$H_1(s)=TH_0(s)=h^*(s).$$
Let $h^{*+}(s)$ (or $H_k^+(s)$) and $h^{*-}(s)$ (or $H_k^-(s)$) denote  $h^*(s)$ (or $H_k(s)$) with BESS energy capacity $E^+_{max}$ and $E^-_{max}$, respectively. We have proved that $h^{*+}(s)\leq h^{*-}(s)$, or equivalently  $H_1^+(s)\leq H^-_1(s)$, if $E^+_{max}\geq E^-_{max}$.  Due to the monotonicity property of contraction mapping,
$$H_k^+(s)\leq H_k^-(s)$$
as long as $H_{k-1}^+(s)\leq H_{k-1}^-(s)$ for all $k$. Taking $k$ to infinity, we have $H^{*+}(s)\leq H^{*-}(s)$ when $E^+_{max}\geq E^-_{max}$. This completes the proof.

\end{proof}

\section{Algorithm to Obtain  $\pi^*_{low}$ and $\pi^*_{high}$}\label{append:algorithm}
The traditional algorithms to solve infinite-horizon dynamic programming problems, e.g., value iteration and policy iteration algorithms, involve iterative steps, where in each iteration, the policy $\pi(i)$ is updated for each system state (i.e., BESS SoC) $i$. In our problem, the state space is continuous in $[0, 1]$. If it is discretized into $N$ levels, i.e., $i\in\{0, \delta, 2\delta, \cdots, 1\}$ where $\delta=\frac{1}{N-1}$, then $N$ optimization problems, one for each $\pi(i)$, need to be solved in each iteration.

Based on the state-invariant property of $\pi^*_{low}$ and $\pi^*_{high}$, the complexity of solving the dynamic programming problem can be greatly reduced. Define $p_{ij}(\pi)=\Pr\{s_{n+1}=j|s_n=i, \pi\}$, which can be calculated from the distributions of $I$, $J$, $q$, and $E_{PFC}$. For any given pair of $\mathbf{d}=(\pi_{low}, \pi_{high})$, we have
\begin{equation}
p_{ij}^\mathbf{d}\doteq p_{ij}(\pi(i))=
\begin{cases}
p_{ij}(\pi_{low}) & i<\pi_{low} \\
p_{ij}(\pi_{high}) & i>\pi_{low} \\
p_{ij}(i) & \pi_{low}\leq i \leq\pi_{high}
\end{cases}
\end{equation}
Let $\mathbf{P}^\mathbf{d}$ be the matrix of $p_{ij}^\mathbf{d}$, and $\mathbf{H}^{\mathbf{d}}$  be the vector of $H^\mathbf{d}(i)$. Likewise, define vector $\mathbf{h}^{\mathbf{d}}$, whose $i^{th}$ entry is $h(i, \pi_{low})$ when $i<\pi_{low}$, $h(i, \pi_{high})$ when $i>\pi_{high}$, and $h(i, i)$ when $\pi_{low}\leq i\leq \pi_{high}$. Then, $\mathbf{H}^{\mathbf{d}}$ can be obtained as the solution of
\begin{equation}
\left(\mathbf{I}-\alpha \mathbf{P}^\mathbf{d}\right)\mathbf{H}^{\mathbf{d}} = \mathbf{h}^\mathbf{d}.
\end{equation}
The optimal $\pi^*_{low}$ and $\pi^*_{high}$ can then be obtained by solving
\begin{equation}\label{eqn:app_opt_pi}
\min_{\pi_{low}, \pi_{high}} \beta^T\left(\mathbf{I}-\alpha \mathbf{P}^\mathbf{d}\right)^{-1}\mathbf{h}^\mathbf{d},
\end{equation}
where $\mathbf{\beta}$ is an arbitrary vector\footnote{The fact that $\pi^*_{low}$ and $\pi^*_{high}$ minimize $H(i)$ for all $i$ implies that the optimal solution to \eqref{eqn:app_opt_pi} is the same for all $\beta$.}. In contrast to the traditional value iteration and policy iteration approaches, no iteration is required here. $\pi^*_{low}$ and $\pi^*_{high}$ can be obtained by solving one optimization problem \eqref{eqn:app_opt_pi} with two scalar variables only.

\end{appendices}

\bibliographystyle{IEEEtran}
\bibliography{ref-freq-control}

\end{document}